\documentclass[journal]{IEEEtran}
\usepackage{amsmath}
\usepackage{amsfonts}
\usepackage{amssymb}
\usepackage{amsthm}
\usepackage{algorithm}
\usepackage{algorithmic}
\usepackage{booktabs}
\usepackage{tabularx}
\usepackage{graphicx}
\usepackage{etoolbox}
\usepackage{balance}
\usepackage{microtype}
\graphicspath{{figures/}}
\usepackage[colorlinks,citecolor=blue,linkcolor=blue]{hyperref}

\newtheorem{proposition}{Proposition}
\newtheorem{corollary}{Corollary}

\newcommand{\R}{\mathbb{R}}
\newcommand{\soft}{\operatorname{soft}}
\newcommand{\diag}{\operatorname{diag}}

\begin{document}

\title{\LARGE  Electromagnetic Twin: From Sparse Measurements to  \\Persistent Wireless Intelligence}

\author{Tuo~Wu,
	Lifeng Mai,
	Jie~Tang, 
	Shihang Lu,
	Kangda Zhi,
	Maged~Elkashlan, 
	M\'erouane~Debbah, \emph{Fellow, IEEE},\\
	Matthew C. Valenti, \emph{Fellow, IEEE},
	Fumiyuki Adachi, \IEEEmembership{Life Fellow,~IEEE},
	 ~Hing Cheung So, \emph{Fellow, IEEE}
\thanks{(\textit{Corresponding author: J. Tang.})}
\thanks{T. Wu and J. Tang are with the School of Electronic and Information Engineering, South China University of Technology, Guangzhou 510640, China (E-mail: $\rm \{wutuo, eejtang\}@scut.edu.cn$). Lifeng Mai is with the Electric Power Research Institute, China Southern Power Grid, Guangzhou 510663, China, and also with the Guangdong Provincial Key Laboratory of Power System Network Security, Guangzhou, Guangdong, China (E-mail: $\rm  mailf@csg.cn$). S. Lu is with the School of Automation and Intelligent Manufacturing (AiM), Southern University of Science and Technology, Shenzhen 518055, China (E-mail: $\rm lush2021@mail.sustech.edu.cn$). K. Zhi is with the School of Electrical Engineering and Computer Science, Technical University of Berlin, 10623 Berlin (E-mail: $\rm k.zhi@tu$-$\rm berlin.de$). M. Elkashlan is with Queen Mary University of London, London, U.K. (E-mail: $\rm maged.elkashlan@qmul.ac.uk$). M. Debbah is with the KU 6G Research Center, Department of Computer and Information Engineering, Khalifa University, Abu Dhabi 127788, UAE (E-mail: $\rm merouane.debbah@ku.ac.ae$). M. C. Valenti is with the Lane Department of Computer Science and Electrical Engineering, West Virginia University, Morgantown, USA (E-mail: $\rm valenti@ieee.org$). F. Adachi is with the International Research Institute of Disaster Science (IRIDeS), Tohoku University, Sendai, Japan (E-mail: $\rm adachi@ecei.tohoku.ac.jp$). H. C. So is with the Department of Electrical Engineering, City University of Hong Kong, Hong Kong (E-mail: $\rm hcso@cityu.edu.hk$). }}

\markboth{}{Wu \MakeLowercase{\textit{et al.}}: Electromagnetic Twin from Sparse Measurements to Persistent Wireless Intelligence}
\maketitle

\begin{abstract}
Radio maps and channel knowledge maps provide reusable propagation knowledge, but a stored map can become locally outdated after persistent changes to doors, partitions, furniture, large equipment, or infrastructure. Motivated by digital-twin state synchronization, we develop an electromagnetic twin that recurrently converts sparse channel measurements into a persistent wireless state, exposes that state to communication queries, and uses its uncertainty to request subsequent measurements. The twin tracks persistent or semi-persistent propagation changes rather than transient human motion or fast fading. Each update is a constrained inverse problem that combines the previous map, a scene graph, and an imperfect physics prior. Measurement consistency with a confidence-calibrated radius and physical gain bounds enforce feasibility, while scene-aware spatial regularization and selective temporal memory preserve propagation boundaries and unchanged regions. Successive convex approximation (SCA), majorization--minimization alternating direction method of multipliers (MM-ADMM), and a low-complexity primal--dual hybrid gradient (LC-PDHG) mode realize the framework at different computational scales, and a local perturbation analysis gives an explicit multi-update tracking recursion. The updated state supports access-point association and codebook beam selection, while change-weighted A-optimal design closes the measurement--update--query loop. Over 500 realizations, the recurrent update remains stable through eight persistent scene events and raises best-beam accuracy from $76.64\%$ to $90.83\%$ as measured-location density grows from $2\%$ to $16\%$.
\end{abstract}


\section{Introduction}
\IEEEPARstart{W}{ireless} networks rely on location-dependent channel information for coverage planning, access-point (AP) association, handover preparation, link adaptation, beam management, and radio-resource allocation. As wireless systems evolve toward sixth-generation (6G) operation, such environment awareness becomes increasingly important for adaptive networks and spatially reconfigurable technologies~\cite{saad6GVision,wuFAS6G}. A dense survey can provide this information at one time instant, but repeating it after every persistent change to the propagation environment---for example, furniture reconfiguration, installation or removal of large equipment, opening or closing a door for an extended period, temporary partitions, an AP outage, or infrastructure modification---requires substantial sounding overhead. These are changes to the site state over a twin-update interval. We do not propose updating the map whenever a person walks through the site or whenever an instantaneous fading coefficient changes; such short-term effects remain the responsibility of conventional channel tracking. In many practical updates, only part of the propagation environment changes and a large fraction of previously verified channel knowledge remains useful. The central question is therefore not simply how to interpolate another radio map, but how to retain valid wireless knowledge, revise the part contradicted by current evidence, and decide what the network should measure next.

Radio maps provide a direct spatial representation of received power, path loss, or another channel quantity and have long supported spectrum cartography and wireless resource management~\cite{biRadioMap,spectrumCartography}. Channel knowledge maps (CKMs) broaden this representation by organizing location-indexed gains, angles, path states, and beam information for reuse~\cite{zengCKM,ckmtutorial}. Data-driven estimators infer unmeasured locations from sparse samples and environmental images using convolutional, graph, completion, Gaussian-regression, and frequency--spatial models~\cite{radiounet,radiogat,deepCompletion,semiParamGP,spectrumReasoning}; channel charting instead learns a geometry-preserving representation directly from channel features~\cite{channelCharting}. Physics-based propagation models and sensor-aided ray tracing inject complementary geometric structure~\cite{sionna,rgbdRay}, while automatic Wi-Fi radio-map maintenance demonstrates the value of incorporating user feedback over time~\cite{radioMapUpdate}. These approaches are valuable components of a site-specific wireless model. However, a map-completion result by itself does not specify which old entries should remain valid, how imperfect scene information should be corrected by new measurements, or how the completed state should participate in the next network action.

We use the term \emph{electromagnetic twin} for this recurrent and operational layer. The twin is not intended to replace a channel knowledge map (CKM), radio map, or propagation simulator. Instead, the CKM stores accumulated channel measurements and metadata, the propagation engine supplies a physical prior, and the radio map is one queryable view of the current state. The electromagnetic twin connects these components through a state transition: it accepts the stored state and sparse current measurements, produces an updated state, supports a communication query, and requests additional measurements when uncertainty remains. This interpretation is compatible with wireless digital-twin, digital-twin-channel, and network-digital-twin architectures~\cite{realtimetwin,dtchannel,rekpool,dt5g,dt6g,networkDT,diten}, but makes the measurement--update--query feedback mechanism explicit enough to analyze and reproduce. Its recurrent operation is summarized by
\begin{equation}
(\widehat{\mathbf g}_{t-1},\mathcal G_{t-1},\mathcal K_{t-1})
\xrightarrow{\;\mathbf y_t,\,\Delta\mathcal G_t\;}
(\widehat{\mathbf g}_{t},\mathcal G_t,\mathcal K_t),
\label{eq:twin_transition}
\end{equation}
where $t$ is the update epoch, $\widehat{\mathbf g}_t$ is the wireless state, $\mathcal G_t$ is the registered scene graph, $\mathcal K_t$ is the CKM evidence memory, $\mathbf y_t$ contains new channel measurements, and $\Delta\mathcal G_t$ denotes registered scene changes. The distinction is functional rather than competitive: the radio map is an output field, the CKM is the evidence memory, and the electromagnetic twin is the recurrent update and query mechanism.

Joint sensing and map completion further reveal that reconstruction and measurement placement are coupled~\cite{jointRadioMap,sensorSelection}, with uncertainty-guided surveying, Gaussian-process placement, and graph sampling providing established acquisition principles~\cite{spectrumSurvey,gpSensorPlacement,graphSampling}. These foundations motivate the closed measurement--update--query loop developed here.

Throughout the paper, a \emph{channel measurement} means a large-scale channel-gain or reference-signal-power sample obtained after conventional transmission of a known pilot waveform, receiver processing, and local averaging. The pilot sequence itself is not optimized here. This distinction matters because the update operates on the resulting measurements and their locations, not on pilot symbols. Realizing the state transition is difficult for three related reasons. First, current measurements may cover only a small fraction of the site, so a measurement-only reconstruction can be unstable away from observed locations. Second, neither source of prior information is uniformly reliable. The previous map is informative in unchanged regions but stale near a persistent environmental event, whereas a path-loss or ray-tracing map may preserve propagation geometry while misplacing an object or underestimating its attenuation. Third, indiscriminate spatial smoothing can erase wall-induced discontinuities, while indiscriminate temporal regularization can suppress a genuine change. A useful update must therefore determine where measurement evidence, physical structure, and temporal memory should dominate without accessing the unknown true change mask.

To address these issues, we formulate one electromagnetic-twin transition as a constrained inverse problem with wireless-interpretable inputs. Sparse channel measurements anchor the current field. A scene graph describes spatial adjacency and weakens information exchange across known propagation boundaries using graph-signal principles~\cite{graphSP}. A confidence field identifies locations at which the stored state is likely to remain valid. An approximate physics map provides the large-scale propagation pattern and is calibrated online by measurement residuals before optimization. The resulting estimator uses reweighted sparsity penalties~\cite{reweightedl1} for a piecewise-smooth correction to the physics map and a sparse innovation relative to the previous twin, while hard constraints enforce measurement consistency and admissible channel-gain limits. Consequently, the model does not choose globally between a data-driven update and a physics-driven update; it combines them spatially according to the available evidence.

The resulting log-sum regularization is nonconvex, and a practical twin must support different computational scales. For a moderate centralized graph, we derive a direct successive convex approximation (SCA) method based on standard MM/SCA principles~\cite{hunterMM,scaTheory}. Each outer iteration constructs a tight convex majorizer and solves its second-order cone representation to high accuracy. For a large or irregular graph, we derive a majorization--minimization alternating direction method of multipliers (MM-ADMM) implementation~\cite{boydADMM,eckstein}. It uses the same outer weights and solves an equivalent splitting of the same convex approximation through a sparse map update, two shrinkage operations, and two projections. Thus, direct SCA is the transparent accuracy reference, whereas MM-ADMM is the structure-exploiting realization of the same update. We additionally derive LC-PDHG for latency-limited operation. It freezes one SCA surrogate per evidence round and uses only sparse operator products and closed-form proximal steps; unlike the first two solvers, it deliberately trades repeated reweighting accuracy for a fixed iteration budget.

An updated state becomes a twin only when it affects subsequent network operation. We demonstrate three interfaces. First, one scalar gain layer per candidate AP supports association using large-scale channel information; the decision is evaluated through spectral efficiency and outage probability without assuming unavailable instantaneous phase. Second, extending the state to a location--beam gain map supports codebook beam selection and quantifies best-beam accuracy and beam-gain loss. Third, the local curvature of the accepted update defines an information matrix for active sensing. A change-weighted A-optimal criterion selects the next channel-measurement batch by balancing uncertainty reduction over the complete state with additional attention to regions indicated by imperfect change information. The new measurements are assimilated by the same constrained update, yielding wireless knowledge that is repeatedly queried and refined.

The contributions are summarized as follows.
\begin{itemize}\setlength{\itemsep}{0pt}\setlength{\parskip}{0pt}
\item We give a concrete electromagnetic-twin state transition that combines sparse-measurement consistency, physical gain bounds, scene-aware spatial corrections, selective temporal memory, and an evidence-calibrated propagation prior. The formulation states explicitly what is stored, what is updated, and what distinguishes recurrent twinning from one-shot map completion.
\item We derive direct SCA for moderate scenes, an equivalent MM-ADMM splitting for large or irregular scenes, and an LC-PDHG mode for latency-limited updates. The first two solve the same sequence of convex approximations; LC-PDHG freezes one surrogate so that its accuracy--complexity tradeoff is explicit rather than hidden inside an implementation choice.
\item We derive a local tracking bound that separates channel-measurement noise, physics-prior mismatch, previous-twin error, persistent state change, and scene-graph mismatch, and we obtain an error-accumulation recursion over repeated updates.
\item We turn the reconstructed state into communication decisions through multi-AP association and codebook beam selection, then close the sensing loop through change-weighted A-optimal measurement acquisition. The evaluation includes a sequence of persistent environmental changes rather than only one map transition.
\end{itemize}


\section{System and Electromagnetic-Twin Update Model}
\subsection{Sparse Channel Measurements}
\label{sec:sparse_measurements}
Consider a service region discretized into $N$ states. A state can represent a two-dimensional grid cell, a three-dimensional voxel, or a location--frequency--beam tuple. The state can also include angle-domain or reconfigurable intelligent surface (RIS) observations: angle estimates support millimeter-wave three-dimensional positioning~\cite{wuMmwavePositioning}, while high-dimensional RIS information can provide localization features even when some elements are faulty~\cite{wuRISLocalization}. The main reconstruction experiments use a scalar gain state, and the beam-management experiment uses a location--beam state.

Let $\mathbf g_t\in\R^N$ denote the channel quantity requested from the twin at update epoch $t$. For the scalar map, $g_{i,t}$ is a dB-valued large-scale channel gain: it includes distance-dependent path loss, shadowing, penetration loss, and any directional gain retained after local averaging, but not instantaneous small-scale phase. At a selected state, a transmitter sends a known reference pilot and the receiver performs conventional channel estimation, converts the result to gain or reference-signal received power, and averages it over the local sounding window. We call the resulting scalar a \emph{channel measurement}; the pilot waveform itself is not an optimization variable. The network obtains measurements at only $M_t\ll N$ states:
\begin{equation}
\mathbf y_t=\mathbf S_t\mathbf g_t+\mathbf n_t,
\label{eq:measurement}
\end{equation}
where $\mathbf y_t\in\R^{M_t}$ is the channel-measurement vector, $\mathbf S_t\in\{0,1\}^{M_t\times N}$ selects the measured states, and $\mathbf n_t$ collects residual channel-estimation error, receiver uncertainty, unresolved fading, and finite-window averaging error. It is distinct from the thermal-noise power used later to evaluate a downlink rate. We require $\|\mathbf n_t\|_2\leq\delta_t$. Unlike one-shot map construction, $\mathbf S_t$ and $\mathbf y_t$ can change after every measurement round. The reconstructed quantity is also bounded by known engineering limits,
$g_{\min}\leq g_{i,t}\leq g_{\max}$; for a dB-valued gain map, these limits can be selected from the receiver dynamic range and link budget. \eqref{eq:measurement} therefore describes a map-level quantity rather than an instantaneous complex channel coefficient. Repeated reference signals or frequency samples make a Gaussian model for $\mathbf n_t$ a practical approximation through averaging; the confidence-calibrated choice of $\delta_t$ is given after Problem~\eqref{prob:p0}.

The twin stores the previous estimate $\widehat{\mathbf g}_{t-1}$ and receives a physics prior $\mathbf p_t$, which may be generated by a path-loss model, ray tracing, or a low-complexity propagation simulator. We use a diagonal confidence matrix
$\mathbf C_t=\diag(c_{1,t},\ldots,c_{N,t})$, where $c_{i,t}\in[0,1]$ is large when state $i$ is believed to remain unchanged. Thus, $c_{i,t}$ increases the cost of changing a reliable old state, whereas the physics prior is emphasized through $\overline{\mathbf C}_t=\mathbf I-\mathbf C_t$ in potentially changed regions. A uniform $\mathbf C_t=c\mathbf I$ is sufficient when no change detector is available.

\subsection{Scene Graph}
Let $\mathcal G_t=(\mathcal V,\mathcal E_t,\mathbf W_t)$ be a graph whose vertex set $\mathcal V=\{1,\ldots,N\}$ contains the channel states, whose edge set $\mathcal E_t$ describes which states may exchange propagation information, and whose diagonal matrix $\mathbf W_t=\diag(\{w_{e,t}\}_{e\in\mathcal E_t})$ stores nonnegative edge weights~\cite{graphSP}. An oriented edge $e=(i,j)$ gives one row of the incidence matrix $\mathbf B_t\in\R^{|\mathcal E_t|\times N}$ satisfying
$(\mathbf B_t\mathbf g)_e=g_i-g_j$. The weight $w_{e,t}\geq0$ expresses scene knowledge: a high weight encourages consistency in open space, while a low weight permits a discontinuity across a wall, blockage, floor boundary, or frequency band. Therefore, the graph translates available geometry into a regularizer without requiring a site-specific neural network.

For a spatial map, one reproducible construction connects four or eight neighboring cells and sets
\begin{equation}
w_{ij,t}=\exp\!\left(-\frac{\|\mathbf x_i-\mathbf x_j\|_2^2}{2\sigma_x^2}
-\kappa_m m_{ij,t}-\kappa_s\|\mathbf s_{i,t}-\mathbf s_{j,t}\|_2^2\right),
\label{eq:graph_weight}
\end{equation}
where $\mathbf x_i$ is the coordinate of state $i$, $\sigma_x$ controls the spatial neighborhood scale, $m_{ij,t}\in\{0,1\}$ indicates whether a registered material boundary intersects the edge, $\kappa_m\geq0$ controls boundary attenuation in the graph, $\mathbf s_{i,t}$ is an optional scene descriptor, and $\kappa_s\geq0$ controls descriptor similarity. Setting $\kappa_m=\kappa_s=0$ recovers a geometry-only graph. A large $\kappa_m$ weakens smoothing across walls without prescribing the gain discontinuity. For a three-dimensional twin, vertical neighbors are added with an altitude-dependent scale. For multi-band or beam maps, cross-frequency or cross-beam edges are appended to the same incidence matrix.

\subsection{Temporal Confidence and Physics Prior}
The previous state should be trusted selectively rather than uniformly. Let $q_{i,t}^{\rm sc}$ be a persistent-scene-change score obtained from building information, camera, lidar, maintenance records, or object-state tracking, and let $q_{i,t}^{\rm ch}$ be a channel-change score inferred from new channel measurements. At a measured vertex, the latter uses the residual relative to the \emph{previous twin},
\begin{equation}
q_{i,t}^{\rm ch}=\min\!\left(1,
\frac{|y_{m,t}-(\mathbf S_t\widehat{\mathbf g}_{t-1})_m|}{\tau_{\rm ch}}
\right),\quad i=\Omega_t(m),
\label{eq:channel_change}
\end{equation}
where $\Omega_t(m)$ maps measurement index $m$ to its graph vertex and $\tau_{\rm ch}>0$ is the gain-change level that saturates the score. Scores at unmeasured vertices can be obtained by one or two normalized graph-diffusion steps. A combined score and confidence are
\begin{equation}
q_{i,t}=\theta q_{i,t}^{\rm sc}+(1-\theta)q_{i,t}^{\rm ch},
\qquad c_{i,t}=\exp(-\alpha q_{i,t}),
\label{eq:combined_confidence}
\end{equation}
with $\theta\in[0,1]$. Thus, the temporal innovation penalty is strong where neither sensing modality indicates a change, and weak where current evidence disagrees with the stored twin.

Let $\mathbf p_t^{(0)}$ denote the raw physics map produced by a path-loss model, low-order ray tracer, neighboring band, or frozen propagation model. Before optimization, the twin uses the current channel measurements to calibrate this map rather than treating it as exact. In contrast to~\eqref{eq:channel_change}, the following residual is taken with respect to the \emph{physics prior}:
\begin{equation}
\mathbf r_t^{p}=\mathbf y_t-\mathbf S_t\mathbf p_t^{(0)},\qquad
\mathbf p_t=\Pi_{\mathcal G}\!\left(\mathbf p_t^{(0)}+
\mathbf R_t\mathbf r_t^{p}\right),
\label{eq:prior_calibration}
\end{equation}
where $\mathbf R_t\in\R^{N\times M_t}$ is a normalized inverse-distance or graph-diffusion operator whose rows are zero at temporally confident vertices. Thus, only a detected persistent-change region receives an evidence-driven correction. In the experiments, each active row uses its five nearest measurements in that region. This inexpensive step turns an approximate physics map into an online prior and never accesses ground-truth values at unmeasured vertices.

The calibrated $\mathbf p_t$ anchors poorly observed changed regions without overriding current measurements. Multiplication by $\overline{\mathbf C}_t$ in \eqref{prob:p0} implements this division: temporal memory dominates when $c_{i,t}$ is high, while the physics term becomes relevant as confidence in the old state decreases. The optimization regularizes the correction field $\mathbf g-\mathbf p_t$, not the absolute map. Sparse measurement residuals can therefore propagate inside a changed scene component without smoothing away the large-scale propagation pattern already supplied by $\mathbf p_t$.

\subsection{Update and Query Time Scales}
The sounding, updating, and querying operations need not share one time scale. Reference pilots may be transmitted every scheduling interval, but their derived channel measurements can be averaged and accumulated until an update trigger is met. A trigger can be elapsed time, an accumulated measurement count, a residual threshold $\max_{m}|y_{m,t}-(\mathbf S_t\widehat{\mathbf g}_{t-1})_m|>\tau_{\rm ch}$, or a registered persistent scene event. Transient human motion and fast fading are averaged or rejected unless their effect persists over the update window. After updating, the network can query a coverage mask, an outage region, the strongest candidate AP, or the best codebook beam. 

We make three modeling assumptions. First, measurements are registered to graph vertices or known interpolation weights, so $\mathbf S_t$ is available. Second, the selected large-scale channel quantity remains approximately constant during one update computation. Third, at least one of temporal memory, physics prior, or measurement evidence anchors each graph component. These assumptions concern map-level tracking, not instantaneous small-scale channel estimation, and are stress-tested in the numerical study.

The distinction among the three representations is now explicit. The CKM stores the samples $\{\mathbf S_\tau,\mathbf y_\tau\}_{\tau\leq t}$ and other indexed channel knowledge. A radio map is a spatial output such as $\widehat{\mathbf g}_t$. The electromagnetic twin connects stored evidence, scene state, an update operator, and subsequent queries in a recurrent loop.

\subsection{Persistent Twin State and Beam-Space Interface}
To make this loop implementable, we distinguish persistent twin state from solver workspace. After update $t$, the network retains
\begin{equation}
\mathcal T_t=\{\widehat{\mathbf g}_t,\mathcal G_t,\mathbf C_t,
\mathbf p_t,\boldsymbol\Sigma_t,\mathcal K_t\},
\label{eq:persistent_twin}
\end{equation}
where $\mathcal K_t$ contains indexed measurements and their time stamps. The map, graph, confidence, calibrated prior, and uncertainty proxy are exposed to the next sensing or communication query. In contrast, SCA weights and ADMM primal--dual variables are numerical workspace. They can be discarded after convergence, or warm-started if the graph and update interval remain similar. This distinction prevents the electromagnetic twin from being interpreted as only an optimizer: $\mathcal T_t$ is the reusable wireless state.

For a scalar state, a query returns a large-scale gain at a location. This supports coverage prediction, AP association, handover preparation, power allocation, and the measurement acquisition developed in Section~\ref{sec:active}. Let $\widehat g_{i,a,t}$ be the completed dB gain from candidate AP $a\in\{1,\ldots,A\}$ to location $i$. Applying the same update independently to the $A$ gain layers gives
\begin{equation}
\begin{aligned}
\widehat a_{i,t}&=\arg\max_a\widehat g_{i,a,t},\\
I_{i,t}&=\sum_{a\ne\widehat a_{i,t}}P_a h_{i,a,t},\\
 \mathcal R_{i,t}&=\log_2\!\left(1+
\frac{P_{\widehat a_{i,t}}h_{i,\widehat a_{i,t},t}}{N_0+I_{i,t}}\right),
\end{aligned}
\label{eq:association_query}
\end{equation}
where $P_a$ is the transmit power in watts, $h_{i,a,t}=10^{g_{i,a,t}/10}$ is the dimensionless true linear channel gain used only to evaluate the selected AP, $I_{i,t}$ is cochannel interference power in watts, $N_0$ is receiver thermal-noise power over the considered bandwidth, and $\mathcal R_{i,t}$ is spectral efficiency in bit/s/Hz. This $N_0$ is different from the map-measurement error $\mathbf n_t$ in~\eqref{eq:measurement}. \eqref{eq:association_query} requires no instantaneous channel phase and is therefore consistent with the reconstructed scalar state.

The scalar map does not determine an arbitrary complex multiuser precoder. It does, however, support practical codebook beamforming after the state is extended to directional gains. Let $g_{i,b,t}$ be the gain at location $i$ for beam $b\in\{1,\ldots,B_q\}$, collect these entries in $\mathbf G_t^{\rm bm}\in\R^{N\times B_q}$, and stack its columns as $\mathbf g_t^{\rm bm}=\operatorname{vec}(\mathbf G_t^{\rm bm})$. A joint location--beam graph uses
\begin{equation}
\mathbf B_t^{\rm bm}=\begin{bmatrix}
\mathbf I_{B_q}\otimes\mathbf B_t^{\rm loc}\\
\mathbf B^{\rm beam}\otimes\mathbf I_N
\end{bmatrix},
\qquad
\widehat b_{i,t}=\arg\max_b\widehat g_{i,b,t},
\label{eq:beam_extension}
\end{equation}
where $\mathbf B_t^{\rm loc}$ connects neighboring locations and $\mathbf B^{\rm beam}$ connects adjacent or correlated codebook beams. Replacing $(\mathbf g_t,\mathbf B_t)$ in~\eqref{prob:p0} by $(\mathbf g_t^{\rm bm},\mathbf B_t^{\rm bm})$ reconstructs a location--beam gain map from sparse beam-sounding measurements. The query $\widehat b_{i,t}$ then selects a realizable analog beam from the codebook. We evaluate best-beam accuracy and the true gain loss $g_{i,b_{i,t}^{\star},t}-g_{i,\widehat b_{i,t},t}$, where $b_{i,t}^{\star}$ is the oracle beam.  

\section{Constrained State Update and Tracking Stability}
The current channel measurements should not merely appear as a soft penalty. Given the noise radius $\delta_t$, they define the measurement-consistency set
\begin{equation}
\mathcal D_t=\left\{\mathbf g:\|\mathbf S_t\mathbf g-\mathbf y_t\|_2
\leq\delta_t\right\}.
\label{eq:data_set}
\end{equation}
Likewise, the physical range defines
$\mathcal G=[g_{\min},g_{\max}]^N$. We estimate a feasible map while penalizing two different structures. Spatial differences of the prior-correction field should be sparse because the physics map already supplies the large-scale geometry; only its model mismatch must be propagated from sparse evidence. Temporal innovations
$\mathbf g-\widehat{\mathbf g}_{t-1}$ should also be sparse because an update should modify only states supported by new evidence or a detected scene change.

The proposed constrained update is
\begin{equation}
\begin{aligned}
\underset{\mathbf g}{\operatorname{min}}\quad
F_t(\mathbf g)={}&\frac{\nu}{2}\|\overline{\mathbf C}_t
(\mathbf g-\mathbf p_t)\|_2^2\\
&+\lambda\sum_{e}w_{e,t}
\log\!\left(1+\frac{|[\mathbf B_t(\mathbf g-\mathbf p_t)]_e|}{\epsilon_g}\right)\\
&+\eta\sum_{i=1}^{N}c_{i,t}
\log\!\left(1+\frac{|g_i-\widehat g_{i,t-1}|}{\epsilon_d}\right)\\
\operatorname{s.t.}\quad
&\mathbf g\in\mathcal D_t\cap\mathcal G,
\end{aligned}
\label{prob:p0}
\end{equation}
where $\lambda,\eta,\nu>0$ and $\epsilon_g,\epsilon_d>0$. The first term limits unsupported departure from the calibrated prior in changed regions. The second promotes a piecewise-smooth prior-correction field, and the third enforces a minimal-change update in temporally reliable regions. In contrast to an unconstrained regularized fit, every feasible solution of \eqref{prob:p0} matches the new measurements to the prescribed error level and lies inside the admissible gain range. For an edge $e$, the notation $[\mathbf B_t(\mathbf g-\mathbf p_t)]_e$ denotes the $e$th entry of the graph-difference vector.

The radius $\delta_t$ has a direct probabilistic interpretation. If the averaged real-valued measurement error satisfies $\mathbf n_t\sim\mathcal N(\mathbf0,\sigma_n^2\mathbf I_{M_t})$, then $\|\mathbf n_t\|_2^2/\sigma_n^2$ is chi-square distributed with $M_t$ degrees of freedom. For a chosen violation probability $\alpha_t\in(0,1)$, we set
\begin{equation}
\delta_t=\sigma_n\sqrt{F_{\chi^2_{M_t}}^{-1}(1-\alpha_t)},
\qquad
\Pr\{\mathbf g_t^\star\in\mathcal D_t\}=1-\alpha_t,
\label{eq:chi_radius}
\end{equation}
where $F_{\chi^2_{M_t}}^{-1}$ is the chi-square quantile and $\mathbf g_t^\star$ is the true map. The experiments use $1-\alpha_t=0.95$ and estimate $\sigma_n$ from repeated local gain measurements. Gaussianity is an approximation justified by averaging, not a defining property of an electromagnetic twin. For heavy-tailed errors, the same update accepts a radius from an empirical quantile or a distribution-free concentration bound. The solvers are initialized by projecting the calibrated prior onto $\mathcal D_t\cap\mathcal G$; if this intersection is empty, a nonnegative measurement slack is penalized and reported rather than silently declaring an infeasible update.

\subsection{Local Tracking Error Across Repeated Updates}
Convergence to a stationary point does not by itself explain when the accepted twin is accurate. We therefore characterize the local sensitivity of an accepted update. Let $\mathbf g_t^\star$ and $\mathbf B_t^\star$ denote the true state and ideal scene incidence operator, let $\mathbf d_t^\star=\mathbf g_t^\star-\mathbf g_{t-1}^\star$ be the persistent environmental change, and define $e_t=\|\widehat{\mathbf g}_t-\mathbf g_t^\star\|_2$. Around an accepted SCA point, smooth the two absolute-value surrogates and denote their nonnegative diagonal local curvatures by $\boldsymbol\Gamma_{g,t}$ and $\boldsymbol\Gamma_{d,t}$. The local KKT Jacobian has the map-space curvature
\begin{equation}
\mathbf H_t^{\rm loc}=\sigma_n^{-2}\mathbf S_t^T\mathbf S_t
+\nu\overline{\mathbf C}_t^T\overline{\mathbf C}_t
+\lambda\mathbf B_t^T\boldsymbol\Gamma_{g,t}\mathbf B_t
+\eta\boldsymbol\Gamma_{d,t}.
\label{eq:tracking_hessian}
\end{equation}

\begin{proposition}[One-step local tracking bound]
\label{prop:tracking_bound}
Assume the active constraints remain unchanged under a sufficiently small perturbation, the ideal-input local update has solution $\mathbf g_t^\star$, and $\mu_t=\lambda_{\min}(\mathbf H_t^{\rm loc})>0$ on its feasible tangent space. Define
\begin{align}
C_{y,t}&=\frac{\sigma_n^{-2}\|\mathbf S_t\|_2}{\mu_t},\qquad
C_{m,t}=\frac{\eta\|\boldsymbol\Gamma_{d,t}\|_2}{\mu_t},\nonumber\\
C_{p,t}&=\frac{1}{\mu_t}\left\|
\nu\overline{\mathbf C}_t^T\overline{\mathbf C}_t+
\lambda\mathbf B_t^{\star T}\boldsymbol\Gamma_{g,t}\mathbf B_t^\star
\right\|_2,\nonumber\\
\mathbf E_{\mathcal G,t}&=
\mathbf B_t^T\boldsymbol\Gamma_{g,t}\mathbf B_t-
\mathbf B_t^{\star T}\boldsymbol\Gamma_{g,t}\mathbf B_t^\star,\nonumber\\
\varepsilon_{\mathcal G,t}&=
\|\mathbf E_{\mathcal G,t}(\mathbf p_t-\mathbf g_t^\star)\|_2.
\label{eq:tracking_constants}
\end{align}
Then the accepted update obeys
\begin{align}
e_t\leq{}&C_{y,t}\delta_t
+C_{p,t}\|\mathbf p_t-\mathbf g_t^\star\|_2
+C_{m,t}e_{t-1}\nonumber\\
&+C_{m,t}\|\mathbf d_t^\star\|_2
+\frac{\lambda}{\mu_t}\varepsilon_{\mathcal G,t}
+o(\varepsilon_t),
\label{eq:tracking_bound}
\end{align}
where $o(\varepsilon_t)$ contains second-order perturbation terms.
\end{proposition}
\begin{proof}
Write the KKT inclusions of the smoothed local surrogate for the actual and ideal inputs and subtract them. Strong monotonicity on the common tangent space gives $\mu_t e_t$ on the left. The measurement perturbation is bounded by $\sigma_n^{-2}\|\mathbf S_t^T\mathbf n_t\|_2\leq\sigma_n^{-2}\|\mathbf S_t\|_2\delta_t$. The prior and graph perturbations give the second and fifth terms in~\eqref{eq:tracking_bound}. Finally,
$\|\widehat{\mathbf g}_{t-1}-\mathbf g_t^\star\|_2
\leq e_{t-1}+\|\mathbf d_t^\star\|_2$ gives the two memory terms. Division by $\mu_t$ proves the result to first order.
\end{proof}


\begin{corollary}[Error accumulation]
\label{cor:error_recursion}
If $C_{m,t}\leq\rho<1$ and the sum of the remaining terms in~\eqref{eq:tracking_bound} is at most $a_t$, then after $T$ updates
\begin{equation}
e_T\leq\rho^T e_0+\sum_{t=1}^{T}\rho^{T-t}a_t
\leq\rho^T e_0+\frac{1-\rho^T}{1-\rho}\max_t a_t.
\label{eq:error_recursion}
\end{equation}
\end{corollary}
The bound does not claim global recovery from arbitrary sparse samples. It states the required local observability condition and shows why measurement density, prior mismatch, graph incompleteness, and temporal memory must be reported together. If $C_{m,t}$ approaches one, the twin should lower temporal confidence or request more measurements rather than repeatedly propagating an uncertain state.

\paragraph{Two computational regimes of one update problem:}
Problem~\eqref{prob:p0} is the only nonconvex estimation model considered in the main development. We distinguish two regimes according to how its convex approximations are computed.

In the \emph{simple-scene regime}, the graph is moderate in size, stored centrally, and updated infrequently enough that a standard conic solve is practical. SCA is applied directly to~\eqref{prob:p0}; each outer iteration solves a second-order cone program (SOCP) to numerical optimality. This regime gives the most transparent algorithm and requires no application-specific splitting.

In the \emph{complex-scene regime}, the map can be large, irregular, or frequently updated, and a generic SOCP solve becomes the computational bottleneck. The statistical model remains~\eqref{prob:p0}. Only its $k$th SCA approximation is rewritten with graph-difference, temporal-innovation, box-consensus, and measurement-residual variables. MM-ADMM then solves that equivalent convex problem by sparse linear algebra and closed-form proximal steps. Thus, ``simple'' and ``complex'' describe computational regimes of the same electromagnetic-twin update, not different objectives fitted to different experiments.

\section{Simple-Scene Update by Direct SCA}
\subsection{Outer MM/SCA Approximation}
This section gives the direct solution of~\eqref{prob:p0}. It is appropriate when the scene graph and map dimension allow each convex approximation to be handled by a standard centralized solver. The same approximation will be reused in Section~\ref{sec:mmadmm} for the complex-scene implementation.
For $s,s^{(k)}\geq0$, concavity gives
\begin{equation}
\log\left(1+\frac{s}{\epsilon}\right)
\leq \log\left(1+\frac{s^{(k)}}{\epsilon}\right)
+\frac{s-s^{(k)}}{\epsilon+s^{(k)}}.
\label{eq:majorization}
\end{equation}
At outer iteration $k$, define the spatial and temporal weights
\begin{align}
a_{e,t}^{(k)}&=\frac{w_{e,t}}
{\epsilon_g+|[\mathbf B_t(\mathbf g^{(k)}-\mathbf p_t)]_e|},
\label{eq:weight_a}\\
b_{i,t}^{(k)}&=\frac{c_{i,t}}
{\epsilon_d+|g_i^{(k)}-\widehat g_{i,t-1}|}.
\label{eq:weight_b}
\end{align}
Removing constants produces the convex surrogate
\begin{equation}
\begin{aligned}
\underset{\mathbf g}{\min}\quad
&\frac{\nu}{2}\|\overline{\mathbf C}_t(\mathbf g-\mathbf p_t)\|_2^2
+\lambda\sum_e a_{e,t}^{(k)}|[\mathbf B_t(\mathbf g-\mathbf p_t)]_e|\\
&\hspace{7mm}+\eta\sum_i b_{i,t}^{(k)}
|g_i-\widehat g_{i,t-1}|\\
\operatorname{s.t.}\quad
&\|\mathbf S_t\mathbf g-\mathbf y_t\|_2\leq\delta_t,
\quad g_{\min}\mathbf 1\preceq\mathbf g\preceq g_{\max}\mathbf 1.
\end{aligned}
\label{prob:pk}
\end{equation}
This step is an MM update because \eqref{eq:majorization} is a tight upper bound, and it is also an SCA update because the concave penalties are replaced by first-order convex surrogates. There is one outer approximation loop, not separate MM and SCA loops.

\subsection{Exact SOCP Form of Each Approximation}
The convexity of \eqref{prob:pk} does not depend on ADMM. To expose a direct solver interface, introduce nonnegative epigraph variables $\mathbf z\in\R^{|\mathcal E_t|}$ and $\mathbf d\in\R^N$, together with a scalar $h$. The spatial and temporal absolute values are represented by
\begin{equation}
-\mathbf z\preceq\mathbf B_t(\mathbf g-\mathbf p_t)\preceq\mathbf z,\qquad
-\mathbf d\preceq\mathbf g-\widehat{\mathbf g}_{t-1}\preceq\mathbf d.
\label{eq:sca_l1_epigraph}
\end{equation}
The quadratic prior term has the second-order cone epigraph
\begin{equation}
\left\|\begin{bmatrix}
\sqrt{2\nu}\,\overline{\mathbf C}_t(\mathbf g-\mathbf p_t)\\ h-1
\end{bmatrix}\right\|_2\leq h+1,
\label{eq:quadratic_soc}
\end{equation}
which is equivalent to $h\geq\frac{\nu}{2}\|\overline{\mathbf C}_t(\mathbf g-\mathbf p_t)\|_2^2$. Therefore, the $k$th approximation can be solved directly as
\begin{equation}
\begin{aligned}
\underset{\mathbf g,\mathbf z,\mathbf d,h}{\operatorname{min}}\quad
&h+\lambda(\mathbf a_t^{(k)})^T\mathbf z
+\eta(\mathbf b_t^{(k)})^T\mathbf d\\
\operatorname{s.t.}\quad
&\eqref{eq:sca_l1_epigraph},\ \eqref{eq:quadratic_soc},\quad
\|\mathbf S_t\mathbf g-\mathbf y_t\|_2\leq\delta_t,\\
&g_{\min}\mathbf1\preceq\mathbf g\preceq g_{\max}\mathbf1,\quad
\mathbf z\succeq\mathbf0,\ \mathbf d\succeq\mathbf0.
\end{aligned}
\label{prob:sca_socp}
\end{equation}
No penalty parameter or variable splitting is needed in this mathematical form. A standard conic solver can return the global optimum of each approximation. In our MATLAB implementation, a tightly converged primal--dual oracle solves the same convex subproblem so that the experiment does not require CVX or a commercial license.

With exact subproblem solutions and $\gamma_k=1$, tightness of \eqref{eq:majorization} gives $F_t(\mathbf g^{(k+1)})\leq F_t(\mathbf g^{(k)})$. Compactness of $\mathcal G$ supplies accumulation points, and standard SCA/MM arguments imply that every accumulation point is stationary under a constraint qualification~\cite{hunterMM,scaTheory}. A detailed descent and stationarity argument is provided in Appendix~\ref{app:mmadmm_proof}. The direct formulation is consequently the cleanest accuracy reference. Its limitation is computational: a generic interior-point implementation introduces $2N+|\mathcal E_t|+1$ primal variables and two second-order cones. Dense worst-case cost is cubic in the problem dimension, although sparse conic solvers perform substantially better.

\section{Complex-Scene Update by MM-ADMM}
\label{sec:mmadmm}
The complex-scene case does not replace~\eqref{prob:p0} with a new reconstruction criterion. It retains the same outer SCA weights~\eqref{eq:weight_a}--\eqref{eq:weight_b}, but solves each convex approximation~\eqref{prob:pk} by sparse linear algebra and closed-form proximal operators. The outer index $k$ therefore updates the nonconvex majorizer, while the inner index $j$ enforces consensus for one fixed $k$. Keeping these two roles separate is essential: MM determines which weighted convex problem is solved, and ADMM determines how that problem is solved.

\subsection{Four-Way ADMM Splitting}
To separate the two nonsmooth penalties and two constraints, introduce
\begin{equation}
\mathbf z=\mathbf B_t(\mathbf g-\mathbf p_t),\quad
\mathbf d=\mathbf g-\widehat{\mathbf g}_{t-1},\quad
\mathbf q=\mathbf g,\quad
\mathbf r=\mathbf S_t\mathbf g-\mathbf y_t.
\label{eq:splitting}
\end{equation}
The variables $\mathbf z$, $\mathbf d$, $\mathbf q$, and $\mathbf r$ respectively represent spatial differences, temporal innovations, a copy projected onto the physical gain box, and the channel-measurement residual. This is called a \emph{splitting} because the single map variable in~\eqref{prob:pk} is copied into four variables so that each nonsmooth term or constraint has a closed-form update. Although four auxiliary variables are used, they form one grouped block opposite $\mathbf g$; the method is therefore a standard two-block ADMM rather than a potentially divergent naive multi-block ADMM~\cite{boydADMM}.

Let $\mathbf A_k=\diag(a_{1,t}^{(k)},\ldots,a_{|\mathcal E_t|,t}^{(k)})$ and
$\mathbf D_k=\diag(b_{1,t}^{(k)},\ldots,b_{N,t}^{(k)})$. Using indicator functions $\iota_{\mathcal G}$ and $\iota_{\mathcal B_t}$, where $\mathcal B_t=\{\mathbf r:\|\mathbf r\|_2\leq\delta_t\}$, subproblem \eqref{prob:pk} is equivalently
\begin{equation}
\begin{aligned}
\min_{\mathbf g,\mathbf z,\mathbf d,\mathbf q,\mathbf r}\;&
\frac{\nu}{2}\|\overline{\mathbf C}_t(\mathbf g-\mathbf p_t)\|_2^2
+\lambda\|\mathbf A_k\mathbf z\|_1
+\eta\|\mathbf D_k\mathbf d\|_1\\
&+\iota_{\mathcal G}(\mathbf q)+\iota_{\mathcal B_t}(\mathbf r)\\
\text{s.t. }\;&\mathbf B_t(\mathbf g-\mathbf p_t)-\mathbf z=\mathbf0,\quad
\mathbf g-\widehat{\mathbf g}_{t-1}-\mathbf d=\mathbf0,\\
&\mathbf g-\mathbf q=\mathbf0,\quad
\mathbf S_t\mathbf g-\mathbf y_t-\mathbf r=\mathbf0.
\end{aligned}
\label{prob:split}
\end{equation}
This form shows why one auxiliary variable is insufficient: the graph and temporal terms require different thresholds, while the box and noise constraints require different projections. More importantly, the splitting does not alter the map-estimation problem.

\begin{proposition}
For any fixed outer iteration $k$, problems~\eqref{prob:pk} and~\eqref{prob:split} have the same feasible $\mathbf g$, the same objective value at corresponding feasible points, and the same set of optimal maps.
\label{prop:equivalence}
\end{proposition}
\begin{proof}
For every feasible $\mathbf g$ of~\eqref{prob:pk}, set $\mathbf z=\mathbf B_t(\mathbf g-\mathbf p_t)$, $\mathbf d=\mathbf g-\widehat{\mathbf g}_{t-1}$, $\mathbf q=\mathbf g$, and $\mathbf r=\mathbf S_t\mathbf g-\mathbf y_t$. The four equalities in~\eqref{prob:split} hold, $\iota_{\mathcal G}(\mathbf q)=0$, and $\iota_{\mathcal B_t}(\mathbf r)=0$. Substitution gives exactly the objective in~\eqref{prob:pk}. Conversely, any finite-objective feasible point of~\eqref{prob:split} satisfies $\mathbf q\in\mathcal G$ and $\mathbf r\in\mathcal B_t$; its equality constraints then imply $\mathbf g\in\mathcal G$ and $\|\mathbf S_t\mathbf g-\mathbf y_t\|_2\leq\delta_t$. Eliminating the auxiliary variables recovers~\eqref{prob:pk}. Hence the two formulations are equivalent in $\mathbf g$.
\end{proof}

Proposition~\ref{prop:equivalence} fixes the role of MM-ADMM. At outer round $k$, the weights are held constant and ADMM converges to the global minimizer of the same convex surrogate used by direct SCA. Only after the inner residual tests are satisfied are the weights recomputed for round $k+1$. Consequently, an accurately solved MM-ADMM sequence inherits the descent and stationarity interpretation of direct SCA; its benefit is scalable computation, not a different optimum created for the complex scene.

With scaled dual variables $\mathbf u_z,\mathbf u_d,\mathbf u_q,$ and $\mathbf u_r$, the scaled augmented Lagrangian, up to dual-only constants, is
\begin{align}
\mathcal L_\rho={}&\frac{\nu}{2}\|\overline{\mathbf C}_t(\mathbf g-\mathbf p_t)\|_2^2
+\lambda\|\mathbf A_k\mathbf z\|_1+\eta\|\mathbf D_k\mathbf d\|_1\nonumber\\
&+\iota_{\mathcal G}(\mathbf q)+\iota_{\mathcal B_t}(\mathbf r)
+\frac{\rho}{2}\|\mathbf B_t(\mathbf g-\mathbf p_t)-\mathbf z+\mathbf u_z\|_2^2\nonumber\\
&+\frac{\rho}{2}\|\mathbf g-\widehat{\mathbf g}_{t-1}-\mathbf d+\mathbf u_d\|_2^2
+\frac{\rho}{2}\|\mathbf g-\mathbf q+\mathbf u_q\|_2^2\nonumber\\
&+\frac{\rho}{2}\|\mathbf S_t\mathbf g-\mathbf y_t-\mathbf r+\mathbf u_r\|_2^2.
\label{eq:aug_lagrangian}
\end{align}
ADMM alternately minimizes \eqref{eq:aug_lagrangian} over the map block and the grouped auxiliary block, followed by dual ascent. Each auxiliary minimization is independent once $\mathbf g$ is fixed.

For penalty $\rho>0$, define
\begin{equation}
\mathbf K_t=\nu\overline{\mathbf C}_t^T\overline{\mathbf C}_t
+\rho(\mathbf B_t^T\mathbf B_t+2\mathbf I+\mathbf S_t^T\mathbf S_t).
\label{eq:K}
\end{equation}
The map update is the sparse linear solve
\begin{align}
\mathbf g^{(j+1)}=\mathbf K_t^{-1}\Big[&\nu\overline{\mathbf C}_t^T
\overline{\mathbf C}_t\mathbf p_t
+\rho\mathbf B_t^T(\mathbf z^{(j)}+\mathbf B_t\mathbf p_t-\mathbf u_z^{(j)})\nonumber\\
&+\rho(\mathbf d^{(j)}+\widehat{\mathbf g}_{t-1}-\mathbf u_d^{(j)})
+\rho(\mathbf q^{(j)}-\mathbf u_q^{(j)})\nonumber\\
&+\rho\mathbf S_t^T(\mathbf r^{(j)}+\mathbf y_t-\mathbf u_r^{(j)})\Big].
\label{eq:gupdate}
\end{align}
\eqref{eq:gupdate} is obtained by setting the gradient of \eqref{eq:aug_lagrangian} with respect to $\mathbf g$ to zero. The term $2\rho\mathbf I$ is contributed by the temporal-consensus and box-consensus equalities. Consequently, $\mathbf K_t$ is symmetric positive definite even when measurements are sparse and the graph Laplacian is singular. No dense inverse is formed. For a moderate fixed graph, a sparse Cholesky factorization is cached. For a large or changing graph, preconditioned conjugate gradient (PCG) uses matrix-vector products with $\mathbf B_t$, $\mathbf B_t^T$, $\mathbf S_t$, and $\mathbf S_t^T$.

A diagonal preconditioner is immediately available:
\begin{equation}
\mathbf M_t=\diag(\mathbf K_t),
\label{eq:jacobi_preconditioner}
\end{equation}
while incomplete Cholesky provides a stronger option when memory permits. The PCG tolerance need not equal the final ADMM tolerance. We use
\begin{equation}
\varepsilon_{\rm CG}^{(j)}=
\min\{\varepsilon_{\max},\chi\|\mathbf r_{\rm pri}^{(j)}\|_2\},
\label{eq:pcg_tolerance}
\end{equation}
so early map solves are inexpensive and become more accurate as consensus is approached.

The four auxiliary updates are separable:
\begin{align}
z_e^{(j+1)}&=\soft\left([\mathbf B_t(\mathbf g^{(j+1)}-\mathbf p_t)
+\mathbf u_z^{(j)}]_e,\frac{\lambda a_{e,t}^{(k)}}{\rho}\right),
\label{eq:zupdate}\\
d_i^{(j+1)}&=\soft\left(g_i^{(j+1)}-\widehat g_{i,t-1}
+u_{d,i}^{(j)},\frac{\eta b_{i,t}^{(k)}}{\rho}\right),
\label{eq:dupdate}\\
\mathbf q^{(j+1)}&=\Pi_{[g_{\min},g_{\max}]^N}
(\mathbf g^{(j+1)}+\mathbf u_q^{(j)}),
\label{eq:qupdate}\\
\mathbf r^{(j+1)}&=\Pi_{\|\cdot\|_2\leq\delta_t}
(\mathbf S_t\mathbf g^{(j+1)}-\mathbf y_t+\mathbf u_r^{(j)}).
\label{eq:rupdate}
\end{align}
Here, $\soft(x,\tau)=\operatorname{sign}(x)\max(|x|-\tau,0)$, the first projection clips each map entry to its physical interval, and the second projection is
$\Pi_{\|\cdot\|_2\leq\delta}(\mathbf v)
=\min(1,\delta/\|\mathbf v\|_2)\mathbf v$.
The thresholds in \eqref{eq:zupdate}--\eqref{eq:dupdate} are edge- and vertex-dependent. A currently large spatial jump receives a small MM weight and is less likely to be removed. Similarly, a large evidence-supported temporal innovation is shrunk less strongly at the next outer iteration. The two projections play a different role: they impose engineering feasibility and are not tuned by $\lambda$ or $\eta$.

The scaled dual updates are
\begin{align}
\mathbf u_z^{(j+1)}&=\mathbf u_z^{(j)}+\mathbf B_t(\mathbf g^{(j+1)}-\mathbf p_t)-\mathbf z^{(j+1)},\nonumber\\
\mathbf u_d^{(j+1)}&=\mathbf u_d^{(j)}+\mathbf g^{(j+1)}
-\widehat{\mathbf g}_{t-1}-\mathbf d^{(j+1)},\nonumber\\
\mathbf u_q^{(j+1)}&=\mathbf u_q^{(j)}+\mathbf g^{(j+1)}-\mathbf q^{(j+1)},\nonumber\\
\mathbf u_r^{(j+1)}&=\mathbf u_r^{(j)}+\mathbf S_t\mathbf g^{(j+1)}
-\mathbf y_t-\mathbf r^{(j+1)}.
\label{eq:dualupdates}
\end{align}

\subsection{Residuals, Adaptive Penalty, and Stopping}
The primal residual stacks the four consensus violations, while a compatible dual residual maps changes of the grouped auxiliary block back to the map space:
\begin{align}
\mathbf r_{\rm pri}^{(j)}={}&[\mathbf B_t(\mathbf g^{(j)}-\mathbf p_t)-\mathbf z^{(j)};
\mathbf g^{(j)}-\widehat{\mathbf g}_{t-1}-\mathbf d^{(j)};\nonumber\\
&\mathbf g^{(j)}-\mathbf q^{(j)};
\mathbf S_t\mathbf g^{(j)}-\mathbf y_t-\mathbf r^{(j)}],
\label{eq:primalres}\\
\mathbf r_{\rm dual}^{(j)}={}&\rho\{\mathbf B_t^T\Delta\mathbf z^{(j)}
+\Delta\mathbf d^{(j)}+\Delta\mathbf q^{(j)}
+\mathbf S_t^T\Delta\mathbf r^{(j)}\}.
\label{eq:dualres}
\end{align}
The inner loop stops when both norms are below absolute-plus-relative tolerances. The absolute part controls numerical consensus in the units of the optimized map; the relative part scales with the current primal and dual magnitudes. Feasibility is also checked directly through
\begin{equation}
v_{\rm data}=\big[\|\mathbf S_t\mathbf g-\mathbf y_t\|_2-\delta_t\big]_+,
\quad
v_{\rm box}=\|\mathbf g-\Pi_{\mathcal G}(\mathbf g)\|_\infty.
\label{eq:violations}
\end{equation}

Residual balancing prevents one consensus direction from converging much more slowly than the other. Every $J_\rho$ iterations, we apply
\begin{equation}
\rho^+=
\begin{cases}
\tau_\rho\rho,&\|\mathbf r_{\rm pri}\|_2>\zeta\|\mathbf r_{\rm dual}\|_2,\\
\rho/\tau_\rho,&\|\mathbf r_{\rm dual}\|_2>\zeta\|\mathbf r_{\rm pri}\|_2,\\
\rho,&\text{otherwise},
\end{cases}
\label{eq:rho_adaptation}
\end{equation}
where $\zeta>1$ and $\tau_\rho>1$. Scaled dual variables are multiplied by $\rho/\rho^+$ to keep the corresponding unscaled multipliers unchanged. Because changing $\rho$ changes $\mathbf K_t$, updates are deliberately infrequent; a new preconditioner is then computed, while all primal and dual iterates are retained.

The outer loop stops when the relative map change and relative objective decrease are both small. To avoid unstable initial reweighting, $\epsilon_g$ and $\epsilon_d$ may follow a continuation schedule from coarse to target values. This high-accuracy mode solves each weighted surrogate to its prescribed tolerance and continues MM until stationarity, which is appropriate for offline planning, periodic network optimization, and benchmark generation.

\subsection{Relation to Direct SCA and Convergence}
\begin{proposition}
Assume $\mathcal D_t\cap\mathcal G$ is nonempty and every convex subproblem~\eqref{prob:split} is solved exactly. Then this algorithm produces the same outer sequence as direct SCA initialized at the same $\mathbf g^{(0)}$, up to nonuniqueness of a surrogate minimizer. Moreover,
$F_t(\mathbf g^{(k+1)})\leq F_t(\mathbf g^{(k)})$, and every accumulation point is stationary for~\eqref{prob:p0} under the standard constraint qualification.
\end{proposition}
\begin{proof}
Proposition~\ref{prop:equivalence} shows that exact ADMM minimization of~\eqref{prob:split} returns an optimizer of~\eqref{prob:pk}. Both log-sum majorizers are tight at $\mathbf g^{(k)}$, so minimizing their sum over the unchanged feasible set yields monotonic descent. Compactness of $\mathcal G$ supplies accumulation points, and the standard MM/SCA argument gives stationarity~\cite{hunterMM,scaTheory}. For a fixed $k$, all auxiliary variables form one convex block opposite $\mathbf g$, so standard two-block ADMM converges to a minimizer of~\eqref{prob:pk}~\cite{boydADMM}.
\end{proof}

For a local graph, $|\mathcal E_t|=\mathcal O(N)$. The two shrinkage steps, box projection, noise-ball projection, and dual updates cost $\mathcal O(N+M_t)$. Matrix $\mathbf K_t$ is sparse and positive definite because it contains $2\rho\mathbf I$. With $K_{\rm CG}$ PCG iterations, the dominant map update costs $\mathcal O(K_{\rm CG}(N+M_t))$ per ADMM iteration. A factorization or preconditioner can be reused across MM iterations and across update rounds unless the graph changes substantially.

Exact inner minimization is not required at early outer rounds. Let $Q_t^{(k)}$ denote the convex objective in~\eqref{prob:pk}. We accept an inexact ADMM candidate only when its feasibility violations in~\eqref{eq:violations} are below tolerance and
\begin{equation}
Q_t^{(k)}(\mathbf g^{(k+1)})
\leq Q_t^{(k)}(\mathbf g^{(k)}).
\label{eq:acceptance}
\end{equation}
If this test fails, the inner ADMM loop continues. The rule prevents an inaccurate complex-scene solve from breaking outer descent. A practical implementation uses loose inner tolerances initially and tightens them as the MM iterates stabilize.

\section{Low-Complexity Extension by LC-PDHG}
\label{sec:lcpdhg}
Direct SCA and MM-ADMM repeatedly update the log-sum majorizer. For a latency-limited node, LC-PDHG instead freezes one surrogate per evidence round. From a feasible reference $\mathbf g_t^{\rm ref}$, define
\begin{equation}
\bar a_{e,t}=\frac{w_{e,t}}{\epsilon_g+|[\mathbf B_t(\mathbf g_t^{\rm ref}-\mathbf p_t)]_e|},
\quad
\bar b_{i,t}=\frac{c_{i,t}}{\epsilon_d+|g_{i,t}^{\rm ref}-\widehat g_{i,t-1}|}.
\label{eq:frozen_weights}
\end{equation}
and let $\overline{\mathbf A}_t$ and $\overline{\mathbf D}_t$ contain these frozen weights. The low-complexity problem is
\begin{equation}
\begin{aligned}
\min_{\mathbf g\in\mathcal G}\;&
f_t(\mathbf g)+\lambda\|\overline{\mathbf A}_t\mathbf B_t(\mathbf g-\mathbf p_t)\|_1
+\eta\|\overline{\mathbf D}_t(\mathbf g-\widehat{\mathbf g}_{t-1})\|_1\\
\text{s.t. }\;&\|\mathbf S_t\mathbf g-\mathbf y_t\|_2\leq\delta_t,
\end{aligned}
\label{prob:lc}
\end{equation}
Unlike Direct SCA and MM-ADMM, LC-PDHG does not recompute the log-sum weights after solving this problem. It therefore converges to the minimizer of one convex frozen surrogate, not generally to a stationary point of the original nonconvex problem~\eqref{prob:p0}. Its purpose is deterministic latency and linear-memory operation; similar reconstruction error in a particular scene does not make the three algorithms mathematically identical.
where $f_t(\mathbf g)=\frac{\nu}{2}\|\overline{\mathbf C}_t(\mathbf g-\mathbf p_t)\|_2^2$. Let $\mathcal L_t\mathbf g=[\mathbf B_t\mathbf g;\mathbf g;\mathbf S_t\mathbf g]$ and let $h$ collect the two weighted $\ell_1$ terms and the measurement-ball indicator, including their affine offsets. With $\boldsymbol\xi=[\boldsymbol\alpha;\boldsymbol\beta;\boldsymbol\gamma]$, one PDHG step is
\begin{align}
\boldsymbol\xi^{(\ell+1)}&=\operatorname{prox}_{\sigma h^*}
\!\left(\boldsymbol\xi^{(\ell)}+\sigma\mathcal L_t\bar{\mathbf g}^{(\ell)}\right),\nonumber\\
\mathbf v_g&=\mathbf g^{(\ell)}-\tau\mathcal L_t^T\boldsymbol\xi^{(\ell+1)},
\label{eq:lcpdhg_compact}\\
g_i^{(\ell+1)}&=\operatorname{clip}\!\left(
\frac{v_{g,i}+\tau\nu(1-c_{i,t})^2p_{i,t}}
{1+\tau\nu(1-c_{i,t})^2},g_{\min},g_{\max}\right).\nonumber
\end{align}
The dual proximal map clips the spatial and temporal blocks to $[-\lambda\bar a_{e,t},\lambda\bar a_{e,t}]$ and $[-\eta\bar b_{i,t},\eta\bar b_{i,t}]$, and applies group shrinkage with threshold $\sigma\delta_t$ to the measurement block~\cite{chambollePock}. Extrapolation uses $\bar{\mathbf g}^{(\ell+1)}=\mathbf g^{(\ell+1)}+\vartheta(\mathbf g^{(\ell+1)}-\mathbf g^{(\ell)})$. For $\tau\sigma<1/(2d_{\max}+2)$, the iterations converge to a minimizer of \eqref{prob:lc}; each costs $\mathcal O(N+|\mathcal E_t|+M_t)$ and uses no factorization~\cite{pockPrecondition}. A fixed budget therefore gives deterministic latency, at the stated cost of omitting repeated reweighting.

\section{Twin-Guided Active Channel Measurement}
\label{sec:active}
The preceding algorithms answer how the twin should update once a set of channel measurements is available. We now ask which $K$ unmeasured vertices should be sounded next. At each selected vertex, a conventional reference pilot is transmitted and processed into one new gain measurement as defined in Section~\ref{sec:sparse_measurements}. This acquisition layer does not change Problem~\eqref{prob:p0} or favor one high-accuracy solver. Direct SCA and MM-ADMM return a stationary map of the full model, whereas LC-PDHG returns the minimizer of its frozen surrogate under a fixed iteration budget. The accepted local geometry is then used only to design $\mathbf S_{t+1}$; the active-measurement experiments use MM-ADMM for the state update.

\subsection{Local Information Surrogate}
Let
\begin{equation}
\widehat{\mathbf z}_t=\mathbf B_t(\widehat{\mathbf g}_t-\mathbf p_t),
\qquad
\widehat{\mathbf d}_t=\widehat{\mathbf g}_t-\widehat{\mathbf g}_{t-1},
\label{eq:active_residuals}
\end{equation}
and let $\widehat{\mathbf a}_t$ and $\widehat{\mathbf b}_t$ be the final outer weights from \eqref{eq:weight_a} and \eqref{eq:weight_b}. The absolute values in the accepted weighted surrogate are nondifferentiable. For sensing design only, we replace $|x|$ locally by
$\phi_{\epsilon_u}(x)=\sqrt{x^2+\epsilon_u^2}$, whose curvature is
$\phi_{\epsilon_u}''(x)=\epsilon_u^2(x^2+\epsilon_u^2)^{-3/2}$.
Define diagonal curvature matrices
\begin{align}
[\boldsymbol\Gamma_{g,t}]_{ee}
&=\widehat a_{e,t}\frac{\epsilon_u^2}
{(\widehat z_{e,t}^2+\epsilon_u^2)^{3/2}},
\label{eq:gamma_g}\\
[\boldsymbol\Gamma_{d,t}]_{ii}
&=\widehat b_{i,t}\frac{\epsilon_u^2}
{(\widehat d_{i,t}^2+\epsilon_u^2)^{3/2}}.
\label{eq:gamma_d}
\end{align}
The local information surrogate is
\begin{align}
\mathbf H_t={}&\nu\overline{\mathbf C}_t^T\overline{\mathbf C}_t
+\lambda\mathbf B_t^T\boldsymbol\Gamma_{g,t}\mathbf B_t
+\eta\boldsymbol\Gamma_{d,t}
\nonumber\\
&+\sigma_n^{-2}\mathbf S_t^T\mathbf S_t
+\epsilon_H\mathbf I,
\qquad
\boldsymbol\Sigma_t=\mathbf H_t^{-1},
\label{eq:local_information}
\end{align}
where $\epsilon_H>0$ ensures positive definiteness. 


\subsection{Change-Weighted A-Optimal Design}
Pure A-optimal design reduces average map uncertainty, while a recurrent twin should also revisit regions in which the stored state is less reliable. We assign vertex $i$ the weight
\begin{equation}
\omega_{i,t}=1+\kappa_c(1-c_{i,t}),
\label{eq:active_weight}
\end{equation}
where $\kappa_c\geq0$. A small confidence $c_{i,t}$ increases the value of uncertainty reduction near an observed scene change, while $\kappa_c=0$ recovers unweighted A-optimal sensing. Importantly, $c_{i,t}$ is computed from the incomplete registered scene and measured residuals in \eqref{eq:combined_confidence}; it does not reveal the true change mask.

Let $s_i=1$ if vertex $i$ is included in the next measurement batch and zero otherwise. To avoid confusion with the scene-edge matrix $\mathbf W_t$, define the query-weight matrix $\boldsymbol\Omega_t=\diag(\boldsymbol\omega_t)$. The active-measurement problem is
\begin{subequations}
\label{prob:active}
\begin{align}
\underset{\mathbf s}{\operatorname{minimize}}\quad
&\operatorname{tr}\!\left[\boldsymbol\Omega_t
\left(\mathbf H_t+\sigma_n^{-2}\diag(\mathbf s)\right)^{-1}\right]
\label{prob:active_obj}\\
\operatorname{subject\ to}\quad
&\mathbf1^T\mathbf s=K,\qquad s_i\in\{0,1\},
\label{prob:active_budget}\\
&s_i=0,\qquad i\in\mathcal M_t,
\label{prob:active_unused}
\end{align}
\end{subequations}
where $\mathcal M_t$ is the set of already measured vertices. The trace in \eqref{prob:active_obj} is the weighted sum of local marginal uncertainties after adding $K$ equal-quality measurements. Constraint \eqref{prob:active_budget} makes the sounding overhead explicit, and \eqref{prob:active_unused} prevents repeated selection within one update cycle. Problem~\eqref{prob:active} is a binary sensor-selection problem~\cite{sensorSelection}; exhaustive search requires $\binom{N-M_t}{K}$ evaluations.

\subsection{Greedy Rank-One Solution and Closed Loop}
The effect of one additional measurement admits a closed form. Suppose $\boldsymbol\Sigma$ is the current design covariance within a greedy batch. Selecting candidate $i$ gives
\begin{equation}
\boldsymbol\Sigma_i^+=\boldsymbol\Sigma-
\frac{\boldsymbol\Sigma\mathbf e_i\mathbf e_i^T\boldsymbol\Sigma}
{\sigma_n^2+\Sigma_{ii}},
\label{eq:sherman_morrison}
\end{equation}
by the Sherman--Morrison identity. Its exact reduction of the weighted A-optimal objective is
\begin{equation}
\Delta_i=
\frac{\mathbf e_i^T\boldsymbol\Sigma\boldsymbol\Omega_t
\boldsymbol\Sigma\mathbf e_i}
{\sigma_n^2+\Sigma_{ii}}.
\label{eq:active_score}
\end{equation}
We select the available vertex with the largest $\Delta_i$, update $\boldsymbol\Sigma$ by \eqref{eq:sherman_morrison}, and repeat until $K$ vertices have been chosen. The numerator rewards a candidate correlated with many high-priority uncertain vertices rather than merely one with a large pointwise residual. The denominator correctly accounts for channel-measurement noise and diminishing returns.

\begin{proposition}
For $\boldsymbol\Omega_t\succ\mathbf0$ and $\boldsymbol\Sigma\succ\mathbf0$, every available measurement selected by \eqref{eq:active_score} gives $\Delta_i>0$. Consequently, the weighted A-optimal uncertainty decreases monotonically within each greedy batch.
\end{proposition}
\begin{proof}
The denominator is positive and the numerator is
$\|\boldsymbol\Omega_t^{1/2}\boldsymbol\Sigma\mathbf e_i\|_2^2>0$.
Taking the weighted trace of \eqref{eq:sherman_morrison} gives the exact reduction $\Delta_i$, proving the claim.
\end{proof}



\section{Numerical Results}
\label{sec:numerical}
\subsection{Setup and Compared Methods}
We implement the algorithms in MATLAB R2024b. The reference system operates at 5.2 GHz over a $24\times18$ m indoor site represented by a $32\times24$ four-neighbor grid. A single access point (AP) at $(3.4,4.5)$ m generates the scalar-map data, and four APs near the site corners are used for association. The large-scale channel gain includes distance-dependent loss, correlated shadowing, and wall penetration, and is normalized from $[-120,-35]$ dB to $[0,1]$ only inside the solver. A persistent equipment reconfiguration changes a localized propagation region. The registered scene and physics prior deliberately use shifted boundaries, missing objects, and attenuations different from the true event; transient pedestrian blockage is not simulated as a twin update.

Unless swept, channel gain is measured at $2\%$ of the vertices with 2 dB standard deviation after local averaging. This sparse setting is intentionally more demanding than routine coverage interpolation and exposes whether the previous state and physics prior are used effectively. Every estimator receives the same measured locations, measured values, and noise realization. Every displayed point is the arithmetic mean over 500 independently seeded channel realizations without post-processing, curve smoothing, or uncertainty shading. MM-ADMM uses five outer rounds and at most 80 inner iterations; direct SCA uses the same five approximations with a tightly converged reference solve. LC-PDHG uses the frozen surrogate in Section~\ref{sec:lcpdhg} and 60 primal--dual iterations. It is reported as the proposed deterministic-latency extension, not as an exact third implementation of the full nonconvex problem.

We set $\nu=0.15$, $\rho=1$, $\epsilon_g=0.035$, $\epsilon_d=0.030$, $\lambda=1.2\times10^{-4}$, and $\eta=1.5\times10^{-4}$. In every realization, $\delta_t$ is the 95\% chi-square radius in~\eqref{eq:chi_radius}; hence its value adapts to both $M_t$ and the assumed measurement-noise variance rather than using a fitted multiplier.

\subsection{Compared Map Estimators}
All baselines use identical channel measurements and the final projection onto
$\mathcal D_t\cap\mathcal G$. For inverse-distance weighting (IDW), let
$\omega_{im}=(\|\mathbf x_i-\mathbf x_m\|_2^2+\epsilon_{\rm I})^{-1}$
over the ten nearest measured vertices; then
$\widetilde g_i^{\rm IDW}=\sum_m\omega_{im}y_{m,t}/\sum_m\omega_{im}$,
where $\epsilon_{\rm I}=2.5\times10^{-3}$. The uncalibrated physics baseline is
$\widehat{\mathbf g}_t^{\rm phy}=\Pi_{\mathcal D_t\cap\mathcal G}(\mathbf p_t^{(0)})$.
The two static scene-aware CKM estimators are a quadratic CKM (Q-CKM) and a total-variation CKM (TV-CKM):
\begin{align}
\widetilde{\mathbf g}_t^{\rm Q}
&=\arg\min_{\mathbf g}\ \frac{\mu_y}{2}
\|\mathbf S_t\mathbf g-\mathbf y_t\|_2^2
+\frac{\epsilon_q}{2}\|\mathbf g\|_2^2\nonumber\\
&\hspace{21mm}+\frac{\mu_q}{2}
\|\mathbf W_t^{1/2}\mathbf B_t\mathbf g\|_2^2,\nonumber\\
\widehat{\mathbf g}_t^{\rm TV}
&\in\arg\min_{\mathbf g\in\mathcal D_t\cap\mathcal G}
\|\mathbf W_t\mathbf B_t\mathbf g\|_1,
\label{eq:static_baselines}
\end{align}
with $(\mu_y,\mu_q,\epsilon_q)=(35,0.65,10^{-5})$; the Q-CKM output is also projected onto the common feasible set. These baselines respectively omit scene/physics/memory, residual calibration, or recurrent memory. Direct SCA and MM-ADMM solve \eqref{prob:p0}, whereas LC-PDHG solves \eqref{prob:lc}; hence all comparisons share evidence and feasibility while isolating the value of recurrent twin information.

The principal recovery metric is the RMSE inside the actually changed region $\mathcal V_t^{\rm ch}$,
\begin{equation}
\operatorname{RMSE}_t^{\rm ch}=
\sqrt{\frac{1}{|\mathcal V_t^{\rm ch}|}
\sum_{i\in\mathcal V_t^{\rm ch}}(\widehat g_{i,t}-g_{i,t}^{\star})^2},
\label{eq:rmse}
\end{equation}
reported in dB. This prevents the much larger unchanged area from hiding update errors. To measure whether an update unnecessarily rewrites stable wireless knowledge, we also use
\begin{equation}
D_t^{\rm un}=|\mathcal V_t^{\rm un}|^{-1}
\sum_{i\in\mathcal V_t^{\rm un}}
|\widehat g_{i,t}-\widehat g_{i,t-1}|.
\label{eq:unchanged_drift}
\end{equation}

\begin{figure*}[t]
\centering
\includegraphics[width=0.76\textwidth]{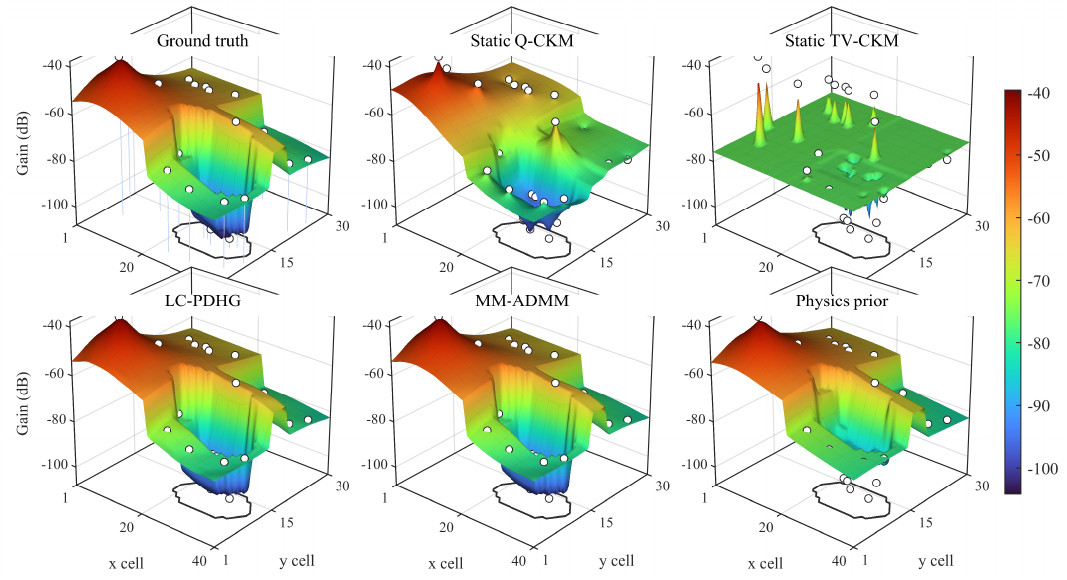}
\caption{Channel-gain surfaces at $2\%$ measurement density.}
\label{fig:maps}
\end{figure*}

\subsection{Map Reconstruction}
Fig.~\ref{fig:maps} shows one representative update on a $40\times30$ grid. Static TV-CKM collapses broad unmeasured areas toward an over-smoothed level and forms localized peaks around isolated measurements, while the uncalibrated physics baseline preserves geometry but underestimates the new equipment-induced attenuation. Static Q-CKM recovers the changed region to 4.875 dB RMSE but rewrites the background because it has no temporal memory. LC-PDHG and MM-ADMM reduce changed-region RMSE to 1.920 and 1.912 dB. Their corresponding full-map RMSE values are 0.692 and 0.695 dB, compared with 4.970 dB for static Q-CKM.

The visual distinction is also temporal. LC-PDHG and MM-ADMM confine the residual calibration to the persistently changed region and retain the old field elsewhere. Static Q-CKM spreads corrections throughout the site, while static TV-CKM collapses large unmeasured areas toward an over-smoothed level. The next experiment quantifies both changed-region recovery and stable-region preservation.

\begin{figure} 
\centering
\includegraphics[width=0.82\columnwidth]{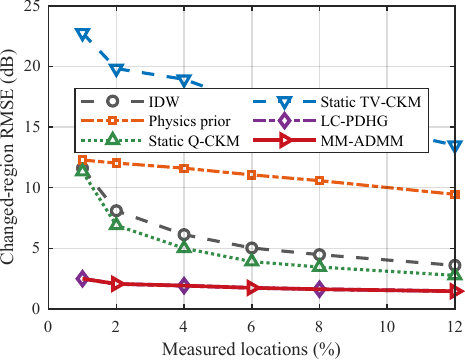}
\caption{Changed-region RMSE versus measurement density (500 realizations).}
\label{fig:pilot}
\end{figure}

\begin{figure} 
\centering
\includegraphics[width=0.82\columnwidth]{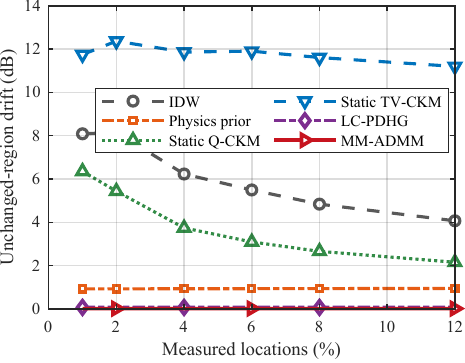}
\caption{Unchanged-region drift versus measurement density.}
\label{fig:drift}
\end{figure}

\subsection{Sparse Evidence and Temporal Stability}
Fig.~\ref{fig:pilot} shows the value of the update loop under sparse evidence. With $1\%$ measured vertices, LC-PDHG and MM-ADMM attain 2.489 and 2.492 dB changed-region RMSE, whereas static Q-CKM, IDW, and the raw physics prior give 11.327, 11.612, and 12.280 dB. At $2\%$, the twin-update values fall to 2.063 and 2.066 dB, compared with 6.894 dB for static Q-CKM. The gain persists as evidence becomes denser: at $12\%$, both displayed twin solvers reach approximately 1.462 dB while static Q-CKM remains at 2.766 dB. The improvement therefore comes from using current measurements to calibrate a reusable propagation state, not merely from changing the numerical optimizer.

Fig.~\ref{fig:drift} supplies the complementary twin-oriented result. At $2\%$ measured vertices, unchanged-region drift is 0.0159 dB for MM-ADMM and 0.0685 dB for LC-PDHG, versus 5.437 dB for static Q-CKM, 12.368 dB for static TV-CKM, and 0.927 dB for the uncalibrated physics map. The drift of both displayed twin solvers remains nearly constant as more measurements arrive. The temporal log-sum term therefore prevents evidence from one local event from unnecessarily rewriting stored wireless knowledge elsewhere.

\begin{figure} 
\centering
\includegraphics[width=0.82\columnwidth]{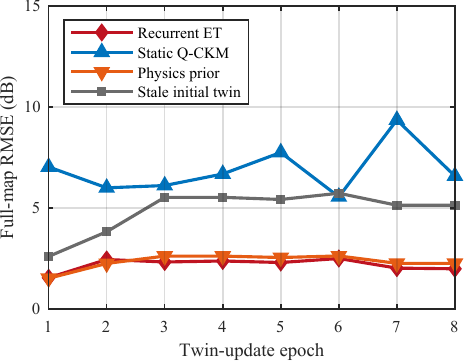}
\caption{Tracking RMSE over eight persistent scene states.}
\label{fig:multiepoch}
\end{figure}

\begin{figure*} 
	\centering
	\includegraphics[width=0.78\textwidth]{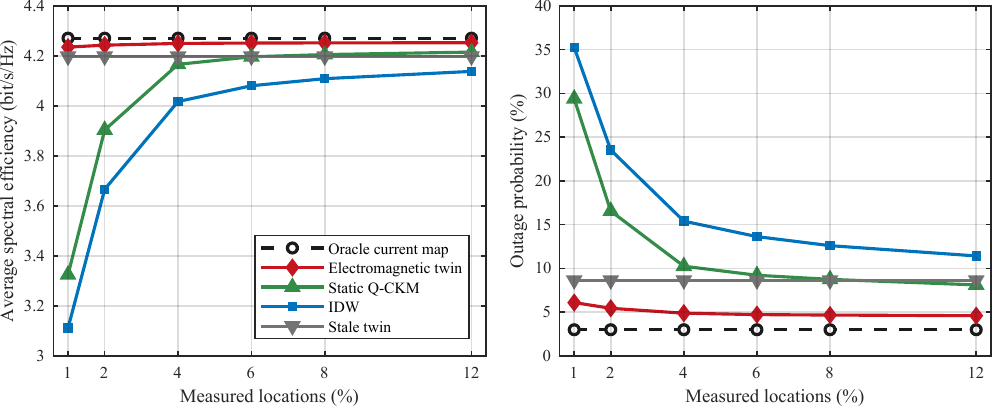}
	\caption{Spectral efficiency and outage versus measurement density.}
	\label{fig:association}
\end{figure*}
\subsection{Persistent Tracking Over Multiple Updates}
Fig.~\ref{fig:multiepoch} tests the state recursion rather than restarting every reconstruction from a fixed map. At every epoch, 4\% of the $20\times14$ grid is measured with 2 dB noise and the recurrent ET accepts its previous output as the next temporal state. The physics engine underestimates attenuation and shifts every registered object boundary; Q-CKM is reconstructed independently from the current measurements; Stale initial twin is never updated. The eight states include two no-change intervals so that a recurrent method is also tested for unnecessary drift.

The recurrent ET remains between 1.54 and 2.49 dB over all 500 sequences and ends at 1.99 dB. In contrast, the static Q-CKM depends strongly on which sparse vertices are observed and ends at 6.58 dB, while the unupdated initial state accumulates 5.13 dB error. The raw physics prior is competitive during the first two changes but cannot use measurements to remove its systematic offset; from the third epoch onward the recurrent ET is consistently more accurate and ends 12.0\% below that prior. This experiment supports the recursion in Corollary~\ref{cor:error_recursion}: accepted updates do not cause unbounded error accumulation, and the two unchanged epochs do not trigger a gratuitous state reset.

\subsection{Using the Twin for Access-Point Association}
Fig.~\ref{fig:association} evaluates the query in \eqref{eq:association_query}, rather than another map-error metric. Four cochannel APs are placed near the corners of the same $32\times24$ grid, and one large-scale gain layer is maintained for each AP. Every realization contains independent correlated shadowing, wall penetration, and an AP-dependent persistent equipment event. All estimators receive the same nested measurement locations and 2 dB measurement noise, and the electromagnetic-twin curve uses MM-ADMM. Each AP transmits at 0 dBm over 20 MHz; the downlink thermal-noise power is computed from $-174$ dBm/Hz and a 7 dB receiver noise figure. Oracle current map selects from the true current layers, whereas Stale twin uses the stored pre-change layers without an update. The other methods select an AP from their completed layers, and the resulting rates are always evaluated on the same true current channels. Outage means $\mathcal R_{i,t}<1$ bit/s/Hz.

At only $2\%$ measured locations, the proposed twin selects the oracle AP at $96.99\%$ of grid locations and achieves 4.243 bit/s/Hz. More importantly, the update reduces outage to $5.43\%$, versus $16.54\%$ for static Q-CKM, $23.52\%$ for IDW, and $8.64\%$ for the stale twin. The latter comparison isolates recurrent updating: temporal memory already provides a strong association state, but retaining it without incorporating persistent environmental change leaves avoidable outages. At the largest measurement density, rate rises to 4.252 bit/s/Hz and outage falls to $4.58\%$. These curves demonstrate a communication use of the scalar twin state without assuming instantaneous channel phase.

\begin{figure} 
\centering
\includegraphics[width=0.82\columnwidth]{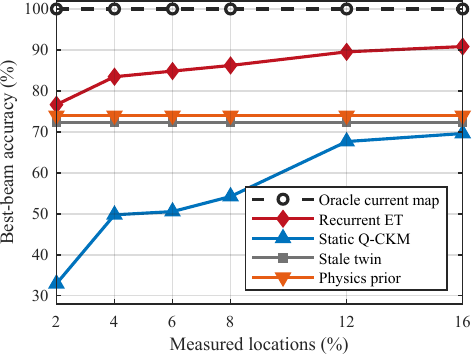}
\caption{Best-beam selection accuracy.}
\label{fig:beam_accuracy}
\end{figure}

\begin{figure} 
\centering
\includegraphics[width=0.82\columnwidth]{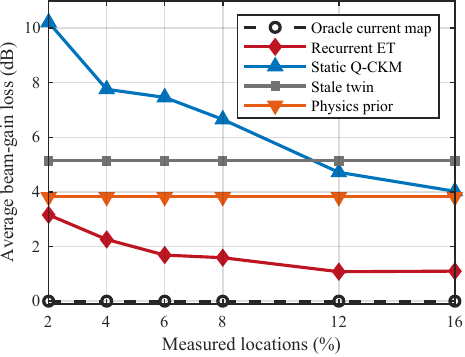}
\caption{Gain loss relative to the best beam.}
\label{fig:beam_loss}
\end{figure}

\subsection{Using the Twin for Codebook Beam Selection}
Figs.~\ref{fig:beam_accuracy} and~\ref{fig:beam_loss} evaluate the location--beam query in~\eqref{eq:beam_extension}. A 5.2 GHz AP with an eight-element uniform linear array serves a $24\times16$ m, $14\times10$ grid using seven analog beams spanning $[-68^\circ,68^\circ]$. At each selected location, all seven reference beams are sounded, so the horizontal axis reports spatial measurement density rather than the fraction of individual location--beam entries. A persistent partition changes both attenuation and dominant departure angle. The physics model underestimates the attenuation and angular displacement, while Stale twin retains the pre-change beam map. Static Q-CKM and the recurrent ET receive identical 2 dB channel measurements; the ET applies the joint location--beam graph in~\eqref{eq:beam_extension} and the MM-ADMM update.

At 2\% measured locations, the recurrent ET selects the true best beam at 76.64\% of grid points and incurs 3.16 dB mean gain loss. At 16\%, its accuracy reaches 90.83\% and its loss falls to 1.10 dB. The static Q-CKM improves from 32.94\% to 69.60\%, but its final 4.02 dB loss remains above the recurrent state. Stale twin and the uncalibrated physics prior remain fixed at 72.31\% and 74.01\% accuracy because neither assimilates the current sounding measurements. Thus, the reconstructed object is not only a scalar coverage image: the same constrained transition supplies a realizable codebook beam decision without claiming recovery of instantaneous complex precoders.

\begin{figure}[t]
\centering
\includegraphics[width=0.82\columnwidth]{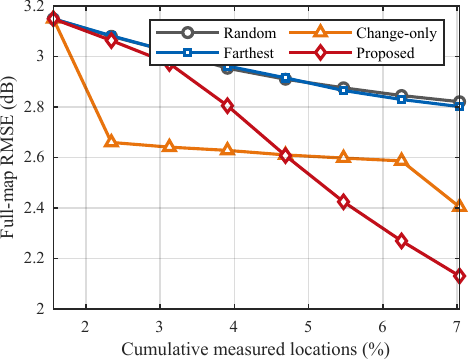}
\caption{Full-map RMSE under active measurement acquisition.}
\label{fig:active_rmse}
\end{figure}

\subsection{Using the Twin to Acquire New Evidence}
Fig.~\ref{fig:active_rmse} tests whether the reconstructed state improves the next measurement decision. Four policies share the same 12 initial measurement locations, channel-measurement noise realization, and MM-ADMM updater. Random selects uniformly without replacement. For the selected set $\mathcal P$, define $d_i(\mathcal P)=\min_{j\in\mathcal P}\|\mathbf x_i-\mathbf x_j\|_2$. The two deterministic baselines select
\begin{equation}
\begin{aligned}
i_{\rm far}&=\arg\max_{i\notin\mathcal P}d_i(\mathcal P),\\
i_{\rm ch}&=\arg\max_{i\notin\mathcal P}
(1-c_{i,t})[0.20+d_i(\mathcal P)].
\end{aligned}
\label{eq:active_baselines}
\end{equation}
Thus, Farthest performs coverage-only exploration, whereas Change-only combines the same imperfect change confidence with spatial spacing but does not propagate map uncertainty. The offset 0.20 prevents a candidate near an existing measurement from receiving an identically zero score when its change confidence is high; distances are normalized by the site diagonal. The proposed policy greedily applies \eqref{eq:active_score} with $\kappa_c=1.5$, $\epsilon_u=0.02$, and $\epsilon_H=2\times10^{-4}$. These values keep the diagonal information matrix well conditioned while giving detected-change vertices at most a $2.5\times$ priority increase. Registered change regions are shifted and $20\%$ incomplete, so neither scene-aware policy receives the true change mask. Six new measurements are added in each of seven batches, giving eight update rounds and 54 final measurements.

\section{Conclusion}
This paper formulated an electromagnetic twin as a concrete measurement--update--query loop. Sparse channel measurements calibrate an imperfect physics prior; a constrained graph update preserves stable wireless knowledge while reconstructing supported persistent changes; and the accepted state drives communication and sensing decisions. In this interpretation, CKM stores reusable evidence, a radio map is one queryable state, and recurrent state synchronization makes that knowledge operational. Direct SCA provides a transparent moderate-scene SOCP reference, MM-ADMM solves an equivalent splitting for larger or irregular graphs, and LC-PDHG supplies a fixed-budget frozen-surrogate mode when latency dominates. The recurrent state remains stable across eight door, partition, furniture, and equipment events.

\appendices
\section{MM-ADMM Descent and Stationarity}
\label{app:mmadmm_proof}
For $\phi_\epsilon(s)=\log(1+s/\epsilon)$, concavity gives the tangent upper bound
\begin{equation}
\phi_\epsilon(s)\leq
\phi_\epsilon(s^{(k)})+\phi_\epsilon'(s^{(k)})(s-s^{(k)}).
\label{eq:tangent_appendix}
\end{equation}
Thus, the surrogate $Q_t^{(k)}$ generated by \eqref{eq:weight_a}--\eqref{eq:weight_b} majorizes $F_t$, touches it at $\mathbf g^{(k)}$, and an exact update obeys
\begin{equation}
F_t(\mathbf g^{(k+1)})
\leq Q_t^{(k)}(\mathbf g^{(k+1)})
\leq Q_t^{(k)}(\mathbf g^{(k)})
=F_t(\mathbf g^{(k)}).
\label{eq:mm_chain}
\end{equation}
Hence $F_t(\mathbf g^{(k)})$ decreases and converges because $\mathcal G$ is compact. For the inner problem, group $\mathbf v=[\mathbf z;\mathbf d;\mathbf q;\mathbf r]$ so that \eqref{prob:split} has the two-block form $\mathbf A\mathbf g+\mathbf E\mathbf v=\mathbf b$, with $\mathbf A=[\mathbf B_t;\mathbf I;\mathbf I;\mathbf S_t]$ and block-diagonal $\mathbf E=-\mathbf I$. The objective is closed, proper, and convex in each grouped block; $\mathbf E$ has full column rank and a saddle point exists whenever \eqref{prob:pk} is feasible. Standard two-block ADMM therefore converges despite the four separable auxiliary variables~\cite{eckstein,boydADMM}. At an accumulation point, the MM weights equal the derivatives of the log penalties, so the limiting surrogate KKT condition is the stationary KKT condition of \eqref{prob:p0}~\cite{hunterMM,scaTheory}. Summable inner errors retain this result; the acceptance test \eqref{eq:acceptance} alone guarantees only descent.

\balance

\end{document}